\documentclass[a4paper,12pt]{article}
\usepackage[utf8]{inputenc}
\usepackage[T2A]{fontenc}
\usepackage{indentfirst}
\usepackage{amsmath}
\usepackage{amsthm}
\usepackage{amsfonts}
\usepackage{amssymb}
\usepackage{graphicx}
\usepackage{setspace}
\usepackage{array}
\usepackage{float}
\usepackage{url}
\usepackage{multirow}
\usepackage{multicol}
\usepackage{enumitem}
\usepackage[title]{appendix}
\usepackage{titling}
\usepackage{longtable}
\usepackage{amsbib}
\usepackage{cite}
\usepackage{geometry}
\usepackage[linesnumbered, ruled, vlined]{algorithm2e}

\newtheorem{theorem}{Theorem}

\newtheorem{definition}{Definition}
\newtheorem{remark}{Remark}

\newcommand{\rand}{\ensuremath{\stackrel{\$}{\leftarrow}}}

\newcommand{\WOTS}{{\rm W-OTS}\ensuremath{^+}}

\providecommand{\keywords}[1]{
	\begin{flushleft}
	\small\textbf{Keywords: } #1
	\end{flushleft}
	}

\providecommand{\affili}[1]{
	\begin{center} 
	\small #1
	\end{center}
	\vspace{.1cm}}

\title{Security analysis of the W-OTS$^+$ signature scheme: \\ Updating security bounds}

\author{\normalfont Mikhail A. Kudinov, Evgeniy O. Kiktenko, and Aleksey K. Fedorov}
\date{}
\begin{document}

\large
\maketitle
\par\vspace{-60pt}
\affili{Russian Quantum Center, Russia \\ QApp, Russia \\ \texttt{mishel.kudinov@gmail.com, e.kiktenko@rqc.ru, akf@rqc.ru}}

\begin{abstract}
	In this work, we discuss in detail a flaw in the original security proof of the W-OTS${^+}$ variant of the Winternitz one-time signature scheme, 
	which is an important component for various stateless and stateful many-time hash-based digital signature schemes.
	We update the security proof for the W-OTS${^+}$ scheme and derive the corresponding security level. 
	Our result is of importance for the security analysis of hash-based digital signature schemes.
\end{abstract}


\section{Introduction}

Many commonly used cryptographic systems are vulnerable with respect to attacks with the use of large-scale quantum computers.
The essence of this vulnerability is the fact that quantum computers would allow solving discrete logarithm and prime factorization problems in polynomial time~\cite{Shor97},
which makes corresponding key sharing schemes and digital signatures schemes breakable. 
At the same time, there exist a number of mathematical operations for which quantum algorithms offer little advantage in speed.
The use of such mathematical operations in cryptographic purposes allows developing quantum-resistant (or post-quantum) algorithms, i.e. cryptographic systems that remain secure under the assumption that the attacker has a large quantum computer.
There are several classes of post-quantum cryptographic systems, which are based on error-correcting codes, lattices, multivariate quadratic equations and hash functions~\cite{Bernstein2017}. 

Among existing post-quantum cryptographic systems, hash-based signature schemes~\cite{Dods} attracted significant attention. 
This is easy to explain since the security of hash-based cryptographic primitives is a subject of extended research activity, and hash functions are actively used in the existing cryptographic infrastructure. 
One of the main components of their security is as follows: For hash functions finding a pre-image for a given output string is computationally hard. 
Up to date known quantum attacks are based on Grover’s algorithm~\cite{Grover}, which gives a quadratic speed-up in the brute-force search. 
Quantum attacks, in this case, are capable to find (i) preimage, (ii) second preimage, and (iii) collision, with time growing sub-exponentially with a length of hash function output.
Moreover, the overall performance of hash-based digital signatures makes them suitable for the practical use.
Several many-time hash-based digital signatures schemes are under consideration for standardization by NIST~\cite{NISThases} and IETF~\cite{WEBSITE:1,WEBSITE:2}. 

We note that still the cryptographic security of hash-based digital signatures is a subject of ongoing debates, so security proofs for such schemes regularly appear (see e.g. \cite{Husling,SPHINCS,Mitigating,SPHINCS+submission,GravitySPHINCS,SPHINCSFramework,WEBSITE:3}).
These studies are partially focused on the security of basic building blocks of many-time hash-based digital signatures, which are one-time signature scheme.
In particular, a variant of the Winternitz signature scheme, which is known as W-OTS${^+}$ is considered.
The original security proof for the W-OTS${^+}$ scheme is presented in Ref.~\cite{Husling}, and the W-OTS${^+}$ scheme is used in XMSS(-MT)~\cite{WEBSITE:2}, SPHINCS~\cite{SPHINCS}, Gravity SPHINCS~\cite{GravitySPHINCS}, and SPHINCS$^+$~\cite{SPHINCS+submission} hash-based digital signatures.
The security of many-time digital signatures obviously depends on the security level of the used one-time signature scheme.

In this work, we study the security of the W-OTS${^+}$ signature scheme.
We identify security flaws in the original security proof for W-OTS${^+}$, which lead to the underestimated level of the security. 
We modify the security analysis of the W-OTS${^+}$ scheme. 

The paper is organized as follows. 
We introduce necessary definitions and notations as well as describe the W-OTS$^{+}$ scheme in Sec.~\ref{sec:prelims}.
In Sec.~\ref{sec:security} we provide a detailed updated security analysis of the W-OTS$^{+}$ and discuss its differences from the previous version.
We conclude in Sec.~\ref{sec:concl}.

\section{Preliminaries} \label{sec:prelims}

\subsection{One-time and many-time hash-based signatures}

The Winternitz one-time signature (W-OTS)~\cite{Merkle,Even} has been introduced as an optimization of the seminal Lamport one-time signature scheme~\cite{Lamport}.
In order to use such one-time signature in practice several its modifications have been discussed. 
In particular, the W-OTS$^{+}$ scheme has received a significant attention in the view of standardization processes, in which one of the candidates is the XMSS signature that uses the W-OTS$^{+}$~\cite{NISThases}.

It order to use hash-based digital signatures in practice one should make them usable for many times. 
In order to do so it is possible to use Merkle trees. 
Using a root of the tree one can authenticate public keys of many one-time signature.
This idea is used in several many-time hash-based signatures based on the W-OTS${^+}$ scheme.
The security of many-time digital signatures clearly depends on the security level of the used one-time signature scheme.
The original security proof for the W-OTS${^+}$ scheme is presented in Ref.~\cite{Husling}.

We note that there are other modifications of the W-OTS scheme (e.g. see~\cite{WEBSITE:1,Johannes}), however they are beyond the scope of the present paper.

\subsection{Definitions and notations}

We start our discussion with introducing basic definitions and notations also used in Ref.~\cite{Husling}.
Let $x \rand X$ denote an element $x$ chosen uniformly at random from some the set $X$.
Let $y \leftarrow {\sf Alg}(x)$ denote an output of the algorithm ${\sf Alg}$ processed on the input $x$.
We write $\log$ instead of $\log_2$ and denote a standard bitwise exclusive or operation with $\oplus$, 
$\lceil \cdot \rceil$ and $\lfloor \cdot \rfloor$ stand for standard ceiling and floor functions.

\begin{definition}[Digital signature schemes]
Let $\mathcal{M}$ be a message space. A digital signature
scheme ${ \sf {Dss} = ({\sf Kg}, Sign, Vf) }$ is a triple of probabilistic polynomial
time algorithms:
	\begin{itemize}
		\item ${\sf Kg}{(1^n)}$ on input of a security parameter $1^n$ outputs a private key ${\sf sk}$ and a public key ${\sf pk}$;
		\item ${\sf Sign(sk},  M)$ outputs a signature $\sigma$ under secret key $\sf sk$ for message $M \in \mathcal{M}$;
		\item ${\sf Vf (pk}, \sigma, M)$ outputs 1 iff $\sigma$ is a valid signature on $M$ under $\sf pk$;
	\end{itemize}
such that $\forall {\sf (pk, sk) \leftarrow {\sf Kg}}{(1^n)}$,$\forall  (M \in \mathcal{M}):{\sf Vf}({\sf pk}, {\sf Sign}({\sf sk}, M),M) = 1$.
\end{definition}

Consider a signature scheme ${\sf Dss}(1^n)$, where $n$ is the security parameter.
A common definition for the security of ${\sf Dss}(1^n)$, which is known as the existential unforgeability under the adaptive chosen message attack (EU-CMA), is defined using the following experiment.
\newline \textbf{Experiment ${\sf Exp}^{\sf EU-CMA}_{{\sf Dss} (1^n)}(\mathcal{A})$}
\begin{itemize}
	\item[] ${\sf (sk, pk) \leftarrow {\sf Kg}}(1^n)$.
	\item[] $(M^{\star},\sigma^{\star}) \leftarrow \mathcal{A}^{\sf sign(sk,\cdot)}({\sf pk})$.
	\item[] $\{(M_i, \sigma_i)\}^{q}_{i=1}$ be the query answers for $\sf {Sign}(sk,\cdot)$.
	\item[] Return 1 iff $\rm{Vf}(pk,\sigma^{\star},M^{\star})=1$ and $M^{\star} \notin \{M_i\}^{q}_{i=1}$.
\end{itemize}
In our work we consider one-time signatures, so the number of allowed quires $q$ is set to 1.

Let $ {\sf Succ}^{\sf EU - CMA}_{{\sf Dss} (1^n)}(\mathcal{A}) = \Pr\big[ {\sf Exp}^{\sf EU - CMA}_{{\sf Dss} (1^n)}(\mathcal{A}) = 1\big]$
be the success probability of an adversary $\mathcal{A}$ in the above experiment.
\begin{definition}[EU-CMA]
	Let $t, n\in \mathbb{N}$, $t = {\rm poly}(n)$, $\rm{Dss}(1^n)$ is a digital signature scheme.
	We call $\rm{Dss}$ ${\rm EU-CMA}$-secure if the maximum success probability 
	${\rm InSec}^{\rm EU-CMA}({\rm Dss}(1^n),t)$ of all possibly probabilistic adversaries $\mathcal{A}$ running in time $\leq t$ is negligible in $n$:
	\begin{equation*}
		{\rm InSec}^{\rm EU - CMA}({\rm Dss}(1^n);t) \stackrel{\rm def}{=} \underset{\mathcal{A}}{\max} \left\{{\sf Succ}^{\sf EU - CMA}_{{\sf Dss} (1^n)}(\mathcal{A})\right\}={\rm negl}(n).
	\end{equation*}
\end{definition}

We then consider proof of the EU-CMA property for the W-OTS$^{+}$ scheme on the basis of the assumption that the scheme is constructed with the function family having some particular properties. 
Let us discuss these required properties in detail.

Consider a function family $\mathcal{F}_n=\{f_k : \{0,1\}^n\rightarrow\{0,1\}^n\}_{k\in\mathcal{K}_n}$, where $\mathcal{K}_n$ is some set.
We assume that it is possible to generate $k\rand\mathcal{K}_n$ and evaluate each function from $\mathcal{F}_n$ for given $n$ in ${\rm poly}(n)$ time.
Then, we require three basic security properties for $\mathcal{F}_n$: (i) it is one-way (OW), (ii) it has the second preimage resistance (SPR) property, and (iii) it has the undetectability (UD) property. 

The success probabilities of an adversary $\mathcal{A}$ against OW and SPR of $\mathcal{F}_n$ are defined as follows:
\begin{multline}
	{\rm Succ}^{\rm OW}_{\mathcal{F}_n}(\mathcal{A})= \\
	\Pr[k \rand \mathcal{K}_n, x \rand \{0,1\}^n, y = f_k(x),x' \leftarrow \mathcal{A}(k,y): y =f_k(x')]
\end{multline}
and
\begin{multline}
    {\rm Succ}^{\rm SPR}_{\mathcal{F}_n}(\mathcal{A})= \\
    \Pr
    [k \rand \mathcal{K}_n, x \rand \{0,1\}^n, x' \leftarrow \mathcal{A}(k,x):
     (x \neq x') \wedge (f_k(x)=f_k(x'))],
\end{multline}
respectively.
By using these notations, we introduce the basic definitions of OW and SPR.

\begin{definition}[One-wayness and second preimage resistance of a function family]
	We call $\mathcal{F}_n$ one-way (second preimage resistant), if the success probability of any adversary $\mathcal{A}$ running in time $\leq t$ against the OW (SPR) of $\mathcal{F}_n$ is negligible:
\begin{equation}
	{\rm InSec^{OW(SPR)}}(\mathcal{F}_n;t) \stackrel{\rm def}{=} \underset{\mathcal{A}}{\rm max}\{ {\rm Succ}^{\rm OW(SPR)}_{\mathcal{F}_n}(\mathcal{A})\} = {\rm negl}(n).
\end{equation}
\end{definition}

To define the UD property we first need to introduce a definition of the (distinguishing) advantage.
\begin{definition}[Advantage] 
	Given two distributions $\mathcal{X}$ and $\mathcal{Y}$ we define the advantage ${\rm Adv}_{\mathcal{X},\mathcal{Y}}(\mathcal{A})$ of an adversary $\mathcal{A}$ in distinguishing between these two distributions as follows:
	\begin{equation}
		{\rm Adv}_{\mathcal{X},\mathcal{Y}}(\mathcal{A}) = |\Pr \big[1\leftarrow \mathcal{A}(\mathcal{X})\big] - \Pr\big[1 \leftarrow\mathcal{A}(\mathcal{Y}) \big]|.
	\end{equation}
\end{definition}

Consider two distributions $\mathcal{D}_{{\rm UD},\mathcal{U}}$ and $\mathcal{D}_{{\rm UD}, \mathcal{F}_n}$ over $\{0,1\}^n \times \mathcal{K}_n$.
Sampling of an element $(u,k)$ from the first distribution $\mathcal{D}_{{\rm UD},\mathcal{U}}$ is realized in the following way: $u\rand \{0,1\}^n$, $k \rand \mathcal{K}_n$.
Sampling of an element $(u,k)$ from the second distribution $\mathcal{D}_{{\rm UD}, \mathcal{F}_n}$ is realized by sampling $k \rand \mathcal{K}_n$ and $x\rand\{0,1\}^n$, and then setting $u=f_k(x)$.
The advantage of an adversary $\mathcal{A}$ against the UD of $\mathcal{F}_n$ is defined as the distinguishing advantage between these distributions:
\begin{equation}
	{\rm Adv}^{{\rm UD}}_{\mathcal{F}_n}(\mathcal{A}) = {\rm Adv}_{\mathcal{D}_{{\rm UD}, \mathcal{U}},\mathcal{D}_{{\rm UD}, \mathcal{F}_n}}(\mathcal{A}).
\end{equation}
\begin{definition}[Undetectability]
	We call $\mathcal{F}_n$ undetectable, if the advantage of any adversary $\mathcal{A}$ against the UD property of $\mathcal{F}_n$ running in time $\leq t$ is negligible:
\begin{equation}
{\rm InSec^{UD}}(\mathcal{F}_n;t) \stackrel{\rm def}{=} \underset{\mathcal{A}}{\rm max}\{{\rm Adv}^{\rm UD}_{\mathcal{F}_n}(\mathcal{A})\} = {\rm negl}(n).
\end{equation}
\end{definition}

\subsection{The W-OTS$^+$ signature scheme}

Here we describe the construction of the W-OTS$^+$ signature scheme.
First of all, we define basic parameters of the scheme.
Let $n \in \mathbb{N}$ be the security parameter, and $m$ be the bit-length of signed messages, that is $\mathcal{M}=\{0,1\}^m$.
Let $w\in \mathbb{N}$ be so-called  Winternitz parameter, which determines a base of the representation that is used in the scheme.
Let us define the following constants:  
\begin{equation}
	l_1 = \left\lceil \frac{m}{\log (w)} \right\rceil, \quad l_2 = \left\lfloor \frac{\log (l_1(w-1))}{\log (w)} \right\rfloor + 1, \quad l=l_1+l_2.
\end{equation}
By using the described above function family $\mathcal{F}_n$, we define a chaining function $c_k^{i}(x,{\bf r})$ for $x\in\{0,1\}^n$,  ${\bf r} = (r_1,\ldots, r_j) \in \{ 0,1\}^{n \times j}$, and $j \geq i\geq 0$ as follows:
\begin{equation}
	c_k^{0}(x,{\bf r})=x, \quad c_k^{i}(x,{\bf r})=f_k(c_k^{i-1}(x,{\bf r}) \oplus r_i) \text{ for } i>0.
\end{equation}
In what follows ${\bf r}_{a,b}$ is a substiring $(r_a,\ldots,r_b)$ of ${\bf r}$ if $b>a$ or it is an empty string otherwise.

Now we are ready to define the basic algorithms of the W-OTS$^+$ scheme.

\textit{Key generation algorithm} (${\sf Kg}(1^n)$) consists of the following steps:
\begin{enumerate}
	\item Sample the values 
	\begin{equation}
		k \rand \mathcal{K}, \quad {\bf r}=(r_1, \ldots, r_{w-1}) \rand \{0,1\}^{n \times (w-1)}.
	\end{equation}
	\item Sample the secret signing key
	\begin{equation}
		{\sf sk} = ({\sf sk}_1, \ldots, {\sf sk}_l) \rand \{0,1\}^{n \times l}.		
	\end{equation}
	\item Compute the public key as follows:
	\begin{equation}
		{\sf pk} = ({\sf pk}_0, {\sf pk}_1, \ldots, {\sf pk}_l) = (({\bf r},k), c_k^{w-1}({\sf sk}_1,{\bf r}), 
	\ldots, c_k^{w-1}({\sf sk}_l,{\bf r})).
	\end{equation}
\end{enumerate}

\textit{Signature algorithm} (${\sf Sign}({\sf sk},M,{\sf {\bf r}})$) consists of the following steps:
\begin{enumerate}
	\item Convert $M$ to the base $w$ representation: $M = (M_1, \ldots, M_{l_1})$ with $M_i\in\{0,\ldots,w-1\}$. 
	\item Compute the checksum $C=\sum_{i=1}^{l_1} (w-1-M_i)$ and its base $w$ representation $C=(C_1,\ldots,C_{l_2})$.
	\item Set $B=(b_1,\ldots, b_l) = M || C$ as the concatenation of the base $w$ representations of $M$ and $C$.
	\item Compute the signature on $M$ as follows:
	\begin{equation}
		\sigma = (\sigma_1, \ldots, \sigma_l) = (c_k^{b_1}({\sf sk}_1,{\bf r}), \ldots, c_k^{b_l}({\sf sk}_l,{\bf r})).
	\end{equation}
\end{enumerate}

\textit{Verification algorithm} (${\sf Vf}({\sf pk},\sigma,M)$) consists of the following steps:
\begin{enumerate}
	\item Compute $(b_1,\ldots,b_l)$ as it is described in steps 1-3 of the signature algorithm.
	\item Do the following comparison:
	\begin{equation}
		{\sf pk}_{i} \stackrel{?}{=} c_k^{w-1-b_i}(\sigma_i,{\bf r}_{b_i+1,w-1}), \quad i\in\{1,\ldots,l\}.
	\end{equation}
	If the comparison holds for all $i$, return 1, otherwise return 0.
\end{enumerate}

We assume that the runtime of all three algorithm is determined by the evaluation of $f_{k}$, while time, which is required for other operations, in negligible.
Thus, the upper bound on the runtime of ${\sf Kg}$, ${\sf Sign}$, ${\sf Vf}$ is given by the value of $lw$.

\section{Security of W-OTS$^+$} \label{sec:security}

\subsection{Security proof}

In this section we consider the security proof of the W-OTS$^+$ scheme. 
The general line of our proof coincides with the one from Ref.~\cite{Husling}.
However there are important differences, which yield another expression for the resulting security value.

\begin{theorem} \label{thm:1} 
	Let $n, w, m \in \mathbb{N}$ and $w,m= {\rm poly}(n)$. Let $\mathcal{F}_n =\{f_k : \{0,1\}^n \rightarrow \{0,1\}^n\}_{k \in \mathcal{K}_n}$ be a one-way, second preimage resistant, and undetectable function family.
	Then, the insecurity of the W-OTS$^+$ scheme against an EU-CMA attack is bounded by
	\begin{multline} \label{eq:ourproof}
		{\rm InSec}^{\rm EU-CMA}(\text{\rm W-OTS}^+(1^n,w,m);t,1) \\
		< lw\cdot\left(
		w \cdot {\rm InSec^{UD}}(\mathcal{F}_n;\widetilde{t}) +
		{\rm InSec^{OW}}(\mathcal{F}_n;\widetilde{t}) + 
		w\cdot {\rm InSec^{SPR}}(\mathcal{F}_n;\widetilde{t})\right)    
	\end{multline}
	with $\widetilde{t} = t + 3 lw +w-2$, where time is given in number of evaluation function from $\mathcal{F}$.
\end{theorem}

\begin{proof}
The proof is by contrapositive.
Suppose there exists an adversary $\mathcal{A}$ that can produce existential forgeries for W-OTS$^+(1^n,w,m)$ scheme by running an adaptive chosen message attack in time $\leq t$ with the success probability $\varepsilon_{\mathcal{A}} \equiv {\rm Succ}_{\text{W-OTS}(1^n, w, m)}^{\rm EU-CMA}(\mathcal{A})$.

Then we are able to construct an oracle machine $\mathcal{M}^{\mathcal{A}}$ that either breaks the OW or SPR of $\mathcal{F}_n$ using the adversary algorithm $\mathcal{A}$.
Consider a pseudo-code description of $\mathcal{M}^{\mathcal{A}}$ in Algorithm~\ref{alg:1} and block scheme in Fig.~\ref{fig:combined}(a).

The algorithm is based on the following idea.
We generate a pair of W-OTS$^{+}$ keys, and then introduce OW and SPR challenges in the $\alpha$th chain, 
where the index of the chain $\alpha$, position of the OW challenge $\beta$, and position of the SPR challenge $\gamma$ are picked up at random [see also Fig.~\ref{fig:combined}(b)].
Then we submit a modified public key ${\sf pk}'$ to $\mathcal{A}$.
The adversary can ask to provide a signature for some message $M$.
If the element $b_{\alpha}$ calculated from $M$ is less than $\beta$, that is it locates below our challenge $y_{c}$, then we are not able to generate a signature and we abort.
Otherwise, we compute the signature $\sigma$ with respect to our modified public key and give it to $\mathcal{A}$.
Finally, we obtain some forged message-signature pair $(M',\sigma')$, and if the forgery is valid then $\sigma'$ eventually contains the solution for the one of our challenges. 
Otherwise $\mathcal{M}^{\mathcal{A}}$ return {\sf fail}.

\begin{algorithm}[H]    \label{alg:1}
\small
	\DontPrintSemicolon
	\SetKwInOut{Input}{Input}\SetKwInOut{Output}{Output}
	\Input{Security parameter $n$, function key $k$, OW challenge $y_c$ and SPR challenge $x_c$.}
	\Output{A value $x$ that is either a preimage of $y_c$ (i.e. $f_k(x)=y$) or a second preimage for $x_c$ under $f_k$ (i.e. $f(x_c)=f(x)$ and $x\neq x_{c}$) or {\sf fail}.}
	Generate W-OTS$^+$ key pair: $({\sf sk},{\sf pk}) \leftarrow {\sf Kg}(1^n)$\;
	Choose random indices $\alpha \rand  \{1, \ldots, l\}, \beta \rand  \{1, \ldots, w-1\}$\;
	\eIf{ $\beta = w-1$}{set ${\bf r}'={\bf r}$}
	{    Choose random index $\gamma \rand  \{\beta+1 \ldots w-1\}$\;
	Set ${\bf r}'={\bf r}$ and replace $r'_{\gamma}$ by $c^{\gamma - \beta - 1}_{k}(y_c, {\bf r}_{\beta+1,w-1}) \oplus x_c$\;}
	Obtain modified public key ${\sf pk}'$ by setting ${\sf pk}_0' = ({\bf r}', k)$, ${\sf pk}_i' = c^{w-1}_{k}({\sf sk}_i,{\bf r}')$ for $1 \leq i \leq l, i \neq \alpha$, and
	${\sf pk}'_{\alpha} = c^{w-1-\beta}_k(y_c, {\bf r}'_{\beta+1,w-1})$\;
	Run $\mathcal{A}^{{\sf Sign(sk,\cdot)}}({\sf pk}')$\;
	\If{ $\mathcal{A}^{{\sf Sign(sk,\cdot)}}({\sf pk}')$ {\rm queries to sign message} $M$}
	{
	Compute $B = (b_1, \ldots, b_l)$ which corresponds to $M$\;
	\If{ $b_{\alpha} < \beta$} { \Return {\sf fail}}
	Generate signature $\sigma$ of $M$ with respect to the modified public key:\newline
		i. Run $\sigma = (\sigma_1, \ldots, \sigma_l) \leftarrow {\sf Sign}(M,{\sf sk,{\bf r}'})$\newline
		ii. Set $\sigma_{\alpha} = c^{b_{\alpha}-\beta}_k(y_c, {\bf r}'_{\beta+1,w-1})$\;
	Reply to the query using $\sigma$\;}
	\eIf{$\mathcal{A}^{{\sf Sign(sk,\cdot)}}({\sf pk}')$ {\rm returns valid} $(\sigma', M')$}{
	Compute $B' = (b_1', \ldots,b_l' )$  which corresponds to $M'$\;
	\uIf {$b_{\alpha}' \geq \beta$}{ \Return {\sf fail}} 
	\uElseIf{ $\beta = w-1$  {\bf or} {$c^{\beta-b_{\alpha}'}_{k}(\sigma'_{\alpha}, {\bf r}'_{b'_{\alpha}+1,w-1}) = y_c$} }{{\bf return} preimage $c^{w-1-b_{\alpha}'-1}_{k}(\sigma'_{\alpha}, {\bf r}'_{b'_{\alpha}+1,w-1}) \oplus r_\beta$}
	\uElseIf{  $x'= c^{\gamma-b_{\alpha}'-1}_k(\sigma'_{\alpha}, {\bf r}'_{b'_{\alpha}+1,w-1}) \oplus {\bf r}_{\gamma} \neq x_c$ {\bf and} 
	$c^{\gamma-b_{\alpha}'}_k(\sigma'_{\alpha}, {\bf r}'_{b'_{\alpha}+1,w-1}) = c^{\gamma - \beta}_{k}(y_c, {\bf r}_{\beta+1,w-1})$} {\Return second preimage $x'=c^{\gamma-b_{\alpha}'-1}_k(\sigma'_{\alpha}, {\bf r}'_{b'_{\alpha}+1,w-1})+r'_{\gamma}$.}
	\Else{ \Return {\sf fail} }
	}
	{ \Return {\sf fail} }
	\caption{$\mathcal{M}^{\mathcal{A}}$}
\end{algorithm}

\begin{figure}
	\centering
	\includegraphics[width=1\linewidth]{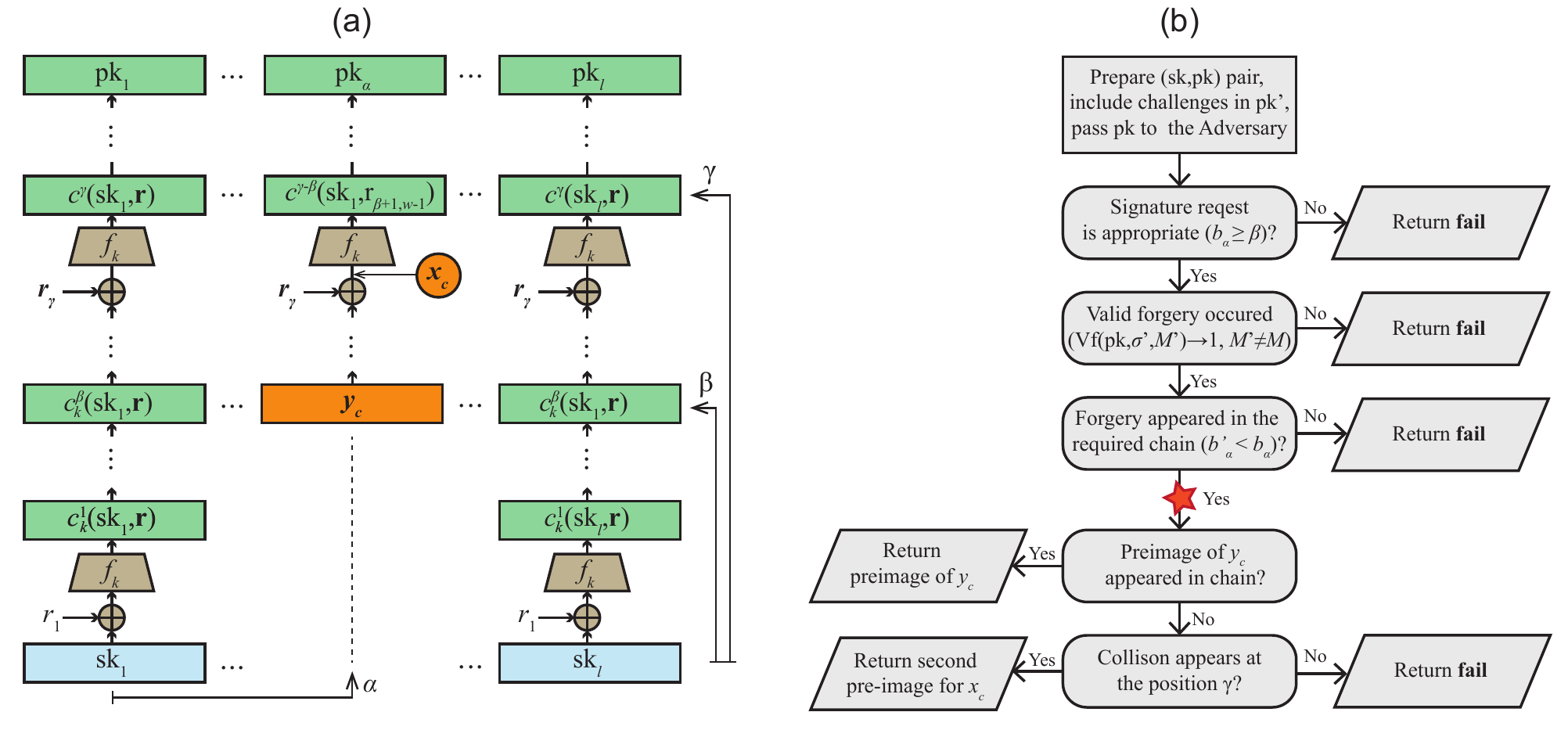}
	\caption{In (a) an introducing image and second pre-image challenges in the public key of the \WOTS~scheme is shown. }
	In (b) the block scheme of $\mathcal{M}^{\mathcal{A}}$ is depicted.
	A bullet marks a point for UD challenge.
	\label{fig:combined}
\end{figure}

We start with computing the success probability of $\mathcal{M}^{\mathcal{A}}$ in solving one of the challenges. 
Let $\widetilde{\epsilon}_{\mathcal{A}}$ be a probability that Algorithm~\ref{alg:1} execution comes to the line 20.
More formally, it can be written as follows:
\begin{equation} \label{eq:fort_forg}
	\widetilde{\epsilon}_{\mathcal{A}}=\Pr[b_{\alpha} \geq \beta \wedge \text{\rm ``Forgery~is~valid''} \wedge b'_{\alpha}<b_{\alpha}],
\end{equation}
where the event ``Forgery is valid'' stands for  $(1\leftarrow{\sf Vf (pk};\sigma' ; M')) \wedge (M'\neq M)$.
We denote the whole event of~Eq.~\eqref{eq:fort_forg} as ``Forgery is fortunate''.

We then can consider two mutually exclusive cases: either
(i) $\beta=w-1$ or the chain started from $\sigma_{\alpha}'$ come to $y_c$ at the $\beta$th level, or (ii) $\beta < w-1$ and the chain started from $\sigma_{\alpha}'$ does not come to $y_c$ at the $\beta$th level.
Let these two case realizing with probabilities $p$ and $(1-p)$ correspondingly conditioned by the event ``Forgery is fortunate''.

In the first case, the adversary $\mathcal{A}$ somehow found a preimage for the $y_c$.
The total probability of this event is upper bounded by ${\rm InSec^{OW}}(\mathcal{F}_n;\widetilde{t})$, so we can write
\begin{equation}\label{eq:OWbound}
	p\cdot \widetilde{\epsilon}_{\mathcal{A}} \leq {\rm InSec^{OW}}(\mathcal{F}_n;\widetilde{t}).
\end{equation}
The time $\widetilde{t}=t+3lw+w-2$ appears as the upper bound on the total running time of $\mathcal{A}$ plus each of the W-OTS$^+$ algorithms ${\sf Kg}$, ${\sf Sign}$, and ${\sf Vf}$ plus preparing $\alpha$th chain in ${\sf pk}'$ (see line 8 in Algorithm~\ref{alg:1}).

In the second case, we have a collision somewhere between $(\beta+1)$th and $(w-1)$ level.
If the collision appears at the level $\gamma$ we obtain the second preimage of $x_c$.
Since the SPR challenge was taken uniformly at random, the value of $r'_{\gamma}$ remains to be a uniformly random variable, therefore there is no way for $\mathcal{A}$ to detect and intentionally avoid the position $\gamma$.
Thus, we obtain the collision at the level $\gamma$ with probability $(w-1-\beta)^{-1}>w^{-1}$ conditioned by the event ``Forgery is fortunate''.
On the other hand, this probability is upper bounded by ${\rm InSec^{SPR}}(\mathcal{F}_n;\widetilde{t})$.
So we have
\begin{equation} \label{eq:SPRbound}
	(1-p) \frac{\widetilde{\epsilon}_{\mathcal{A}}}{w} < {\rm InSec^{SPR}}(\mathcal{F}_n;\widetilde{t}).
\end{equation}
Again, the time $\widetilde{t}=t+3lw+w-2$ appears as the upper bound on the total running time of our algorithm.

By combining Eq.~\eqref{eq:OWbound} and Eq.~\eqref{eq:SPRbound}, we obtain the following expression:
\begin{equation} \label{eq:Insecs}
	\widetilde{\epsilon}_{\mathcal{A}} < {\rm InSec^{OW}}(\mathcal{F}_n;\widetilde{t})+w\cdot {\rm InSec^{SPR}}(\mathcal{F}_n;\widetilde{t}).
\end{equation}

In the remainder of the proof we derive a lower bound for $\widetilde{\epsilon}_{\mathcal{A}}$ as the function of $\epsilon_{\mathcal{A}}$.
We note that in general $\mathcal{A}$ may behave in a `nasty' way making $\widetilde{\epsilon}_{\mathcal{A}}\ll \epsilon_{\mathcal{A}}$ e.g. by always asking to sign `bad' messages with $b_{\alpha}<\beta$ or avoiding forgeries in `good' positions $b'_{\alpha}>b_{\alpha}$.
In other words, the algorithm may avoid crossing the point shown in Fig.~\ref{fig:combined}(a).
This behaviour of $\mathcal{A}$ means that it can somehow reveal the challenge position from the modified public key ${\sf pk}'$.
We below consider the strategy of using this possible ability of $\mathcal{A}$ to break UD property.

Consider two distributions $\mathcal{D}_{\mathcal{M}}$ and $\mathcal{D}_{{\sf Kg}}$ over $\{1, \ldots, w-1\} \times \{0,1\}^n \times \{0,1\}^{n\times(w-1)} \times \mathcal{K}_{n}$.
An element $(\beta,u,{\bf r},k)$ is obtained from $\mathcal{D}_{\mathcal{M}}$ by generating all subelements $\beta$, $u$, ${\bf r}$, and $k$ uniformly at random from the corresponding sets.
At the same time, an element $(\beta,u,{\bf r},k)$ is obtained from $\mathcal{D}_{{\sf Kg}}$ by generating $\beta$, ${\bf r}$, and $k$ uniformly at random, but setting $u=c_{k}^{\beta}(x,{\bf r})$ with $x\rand \{0,1\}^{n}$.
One can see that $\mathcal{D}_{{\sf Kg}}$ corresponds to the generation of elements in W-OTS$^{+}$ signature chain from the secret key element up to the $\beta$th level.

Consider a pseudocode of Algorithm~\ref{alg:2} of a machine $\mathcal{M}'^{\mathcal{A}}$ taking the security parameter $n$ and an element from either $\mathcal{D}_{\mathcal{M}}$ or $\mathcal{D}_{{\sf Kg}}$ as input.
One can see that the operation of $\mathcal{M}'^{\mathcal{A}}$ is very similar to the operation of $\mathcal{M}^{\mathcal{A}}$.

Given an input $(\beta,u,{\bf r},k)$ from $\mathcal{D}_{\mathcal{M}}$, $\mathcal{M}'^{\mathcal{A}}$ sets $y_{c}=u$ and then works exactly as $\mathcal{M}$ up to line 19 of the Algorithm~\ref{alg:1}.
If the event ``Forgery is fortunate'' happens, then $\mathcal{M}'^{\mathcal{A}}$ returns 1. Otherwise, it returns 0.
So given an input $(\beta,u,{\bf r},k)$ from $\mathcal{D}_{\mathcal{M}}$, $\mathcal{M}'^{\mathcal{A}}$  outputs 1 with probability $\widetilde{\epsilon}_{\mathcal{A}}$.

\begin{algorithm}[h]   \label{alg:2}
\small
	\DontPrintSemicolon
	\SetKwInOut{Input}{Input}\SetKwInOut{Output}{Output}
	\Input{Security parameter $n$, a sample $(\beta,u,{\bf r},k)$.} 
	\Output{0 or 1.}
	Generate W-OTS$^+$ key pair: $({\sf sk},{\sf pk}) \leftarrow {\sf Kg}(1^n)$ taking bitmasks from ${\bf r}$ and a function for chain $f_{k}$ instead of random ones\;
	Choose random index $\alpha \rand  \{1, \ldots, l\}$\;
	Obtain modified public key ${\sf pk}'$ by setting ${\sf pk}_0' = ({\bf r}, k)$, ${\sf pk}_i' = c^{w-1}_{k}({\sf sk}_i,{\bf r}')$ for $1 \leq i \leq l, i \neq \alpha$, and
	${\sf pk}'_{\alpha} = c^{w-1-\beta}_k(u, {\bf r}'_{\beta+1,w-1})$\;
	Run $\mathcal{A}^{{\sf Sign(sk,\cdot)}}({\sf pk}')$\;
	\If{ $\mathcal{A}^{{\sf Sign(sk,\cdot)}}({\sf pk}')$ {\rm queries to sign message} $M$}
	{
	Compute $B = (b_1, \ldots, b_l)$ which corresponds to $M$\;
	\If{ $b_{\alpha} < \beta$} { \Return 0}
	Generate signature $\sigma$ of $M$ with respect to the modified public key:\newline
		i. Run $\sigma = (\sigma_1, \ldots, \sigma_l) \leftarrow {\sf Sign}(M,{\sf sk,{\bf r}'})$\newline
		ii. Set $\sigma_{\alpha} = c^{b_{\alpha}-\beta}_k(y_c, {\bf r}'_{\beta+1,w-1})$\;
	Reply to the query using $\sigma$\;}
	\If{$\mathcal{A}^{{\sf Sign(sk,\cdot)}}({\sf pk}')$ {\rm returns valid} $(\sigma', M')$}{
	Compute $B' = (b_1', \ldots,b_l' )$  which corresponds to $M'$\;
	\If {$b_{\alpha}' \geq \beta$}{\Return 0}
	\Else{\Return 1} }
	\caption{$\mathcal{M}'^{\mathcal{A}}$}
\end{algorithm}

Let us consider the behavior of $\mathcal{M}'^{\mathcal{A}}$ given an input from $\mathcal{D}_{{\sf Kg}}$.
In this case $\mathcal{A}$ obtains a fair W-OTS$^{+}$ public key.
The probability that $\mathcal{M}'^{\mathcal{A}}$ outputs 1 is thus given by
\begin{multline}\label{eq:neweps1}
	 \widehat{\epsilon}_{\mathcal{A}} \equiv \Pr[b_{\alpha} \geq \beta \wedge \text{\rm ``Forgery~is~valid''} \wedge b'_{\alpha}<b_{\alpha}] \\= \epsilon_{\mathcal{A}} \cdot \Pr[b_{\alpha} \geq \beta \wedge b'_{\alpha}<b_{\alpha} |  \text{\rm ``Forgery~is~valid''}]\\
	\geq \epsilon_{\mathcal{A}} \cdot \Pr[b_{\alpha} = \beta \wedge b'_{\alpha}<b_{\alpha} |  \text{\rm ``Forgery~is~valid''}].
\end{multline}
Here we used the fact that in the considered case $\Pr[ \text{\rm ``Forgery~is~valid''} ]=\epsilon_{\mathcal{A}}$.
Then we can write
\begin{multline} \label{eq:neweps2}
	\Pr[b_{\alpha} = \beta \wedge b'_{\alpha}<b_{\alpha} |  \text{\rm ``Forgery~is~valid''}] \\ =
	\Pr[b_{\alpha} = \beta |  \text{\rm ``Forgery~is~valid''}] \cdot
	\Pr[b'_{\alpha}<b_{\alpha}  | b_{\alpha} = \beta \wedge  \text{\rm ``Forgery~is~valid''}] 
\end{multline}
and consider each term of the RHS in detail.
Let $X$ be a random variable equal to a number of elements in the requested W-OTS$^{+}$ signature $\sigma$ which lie above the zero level, which is conditioned by the fact the the forgery produced by $\mathcal{A}$ is valid (if one gives $\sigma$ to  $\mathcal{A}$).
More formally we define $X$ as follows:
\begin{equation}
	X = |\{i: 1\leq i \leq l, b_i>0\}|\quad
	\text{conditioned by ``Forgery~is~valid''}.
\end{equation}
Since $\alpha$ and $\beta$ are chosen at random from the sets $\{1,\ldots,l\}$ and $\{1,\ldots,w-1\}$ we have
\begin{equation}\label{eq:aboutX1}
	\Pr[b_{\alpha} = \beta |  \text{\rm ``Forgery~is~valid''}] = \frac{X}{l(w-1)}>\frac{X}{lw}.
\end{equation}
Then, since the forged message $M'$ has at least one element in its signature $\sigma'$ which went down through its chain compared to the signature $\sigma$, and this element is certainly among $X$ elements, we have
\begin{equation} \label{eq:aboutX2}
	\Pr[b'_{\alpha}<b_{\alpha}  | b_{\alpha} = \beta \wedge  \text{\rm ``Forgery~is~valid''}]\geq \frac{1}{X}.
\end{equation}
Taking together Eqs.~\eqref{eq:neweps1}, \eqref{eq:aboutX1}, \eqref{eq:aboutX2}, and putting the result into Eq.~\eqref{eq:neweps1} we obtain
\begin{equation} \label{eq:boundeps}
	\widehat{\epsilon}_{\mathcal{A}} > \frac{\epsilon_{\mathcal{A}}}{lw}.
\end{equation}

By the definition, the advantage of distinguishing $\mathcal{D}_{\mathcal{M}}$ $\mathcal{D}_{{\sf Kg}}$ by $\mathcal{M}'^{\mathcal{A}}$ is given by
\begin{equation}\label{eq:absval}
	   {\rm Adv}_{\mathcal{D}_{\mathcal{M}},\mathcal{D}_{{\sf Kg}}}(\mathcal{M}'^{\mathcal{A}}) = |\widetilde{\epsilon}_{\mathcal{A}}-\widehat{\epsilon}_{\mathcal{A}}|.
\end{equation}
Using the obtained bound~\eqref{eq:boundeps} and expanding absolute value in Eq.~\eqref{eq:absval} we come to the following upper bound on ${\epsilon}_{\mathcal{A}}$:
\begin{equation} \label{eq:epsineq}
	\epsilon_{\mathcal{A}}<lw\cdot\left({\rm Adv}_{\mathcal{D}_{\mathcal{M}},\mathcal{D}_{{\sf Kg}}}(\mathcal{M}'^{\mathcal{A}})+\widetilde{\epsilon}_{\mathcal{A}}\right).
\end{equation}

The remaining step is to derive an upper bound of ${\rm Adv}_{\mathcal{D}_{\mathcal{M}},\mathcal{D}_{{\sf Kg}}}(\mathcal{M}'^{\mathcal{A}})$ using the maximal possible insecurity level of the UD property.
For this purpose we employ the hybrid argument method.
First, we note that 
\begin{equation} \label{eq:betasum}
	{\rm Adv}_{\mathcal{D}_{\mathcal{M}},\mathcal{D}_{{\sf Kg}}}(\mathcal{M}'^{\mathcal{A}}) = \sum_{\beta'=1}^{w-1} \frac{1}{w-1}{\rm Adv}_{\mathcal{D}^{\beta=\beta'}_{\mathcal{M}},\mathcal{D}^{\beta=\beta'}_{{\sf Kg}}}(\mathcal{M}'^{\mathcal{A}}),
\end{equation}
where $\mathcal{D}^{\beta=\beta'}_{\mathcal{M}}$ and $\mathcal{D}^{\beta=\beta'}_{{\sf Kg}}$ denote distributions with fixed first subelement $\beta=\beta'$.
Expression~\eqref{eq:betasum} leads to the fact that there must exist at least one value $\beta^\star$ such that
\begin{equation}\label{eq:predhybr}
	{\rm Adv}_{\mathcal{D}^{\beta=\beta^{\star}}_{\mathcal{M}},\mathcal{D}^{\beta=\beta^{\star}}_{{\sf Kg}}}(\mathcal{M}'^{\mathcal{A}}) \geq  {\rm Adv}_{\mathcal{D}_{\mathcal{M}},\mathcal{D}_{{\sf Kg}}}(\mathcal{M}'^{\mathcal{A}}) .
\end{equation}
Then we define a sequence of distributions $\{\mathcal{H}_{i}\}_{i=0}^{\beta^{\star}}$ over $\{1, \ldots, w-1\} \times \{0,1\}^n \times \{0,1\}^{n\times(w-1)} \times \mathcal{K}_{n}$, such that an element $(\beta,u,{\bf r},k)$ is generated from $\mathcal{H}_{i}$ by setting
\begin{equation}
	\beta = \beta^{\star}, \quad x\rand\{0,1\}^{n}, \quad u=c_k^{\beta^{\star} - i}(x, {\bf r}_{j+1,w-1}),
\end{equation}
and sampling ${\bf r}$ and $k$ uniformly at random from the corresponding spaces.
One can see that $\mathcal{H}_{0}$ and $\mathcal{H}_{\beta^{\star}}$ coincide with $\mathcal{D}_{{\sf Kg}}^{\beta=\beta^{\star}}$ and $\mathcal{D}^{\beta=\beta^{\star}}_{\mathcal{M}}$, correspondingly.
So, Eq.~\eqref{eq:predhybr} can be rewritten as follows:
\begin{equation}
	{\rm Adv}_{\mathcal{H}_{\beta^{\star}},\mathcal{H}_{0}}(\mathcal{M}'^{\mathcal{A}}) \geq  {\rm Adv}_{\mathcal{D}_{\mathcal{M}},\mathcal{D}_{{\sf Kg}}}(\mathcal{M}'^{\mathcal{A}}) .
\end{equation}
The triangular inequality yields the fact that there must exist two consecutive distributions $\mathcal{H}_{i^{\star}}$ and $\mathcal{H}_{i^{\star}+1}$ with $0\leq i^{\star} <\beta^{\star}$ such that
\begin{equation}\label{eq:ineqs}
	{\rm Adv}_{\mathcal{H}_{i^{\star}},\mathcal{H}_{i^{\star}+1}}(\mathcal{M}'^{\mathcal{A}})
	\geq \frac{1}{\beta^{\star}}
	{\rm Adv}_{\mathcal{D}_{\mathcal{M}},\mathcal{D}_{{\sf Kg}}}(\mathcal{M}'^{\mathcal{A}})  
	> \frac{1}{w}
	{\rm Adv}_{\mathcal{D}_{\mathcal{M}},\mathcal{D}_{{\sf Kg}}}(\mathcal{M}'^{\mathcal{A}}).
\end{equation}
We are ready to construct our final machine $\mathcal{B}^{ \mathcal{M}'^{\mathcal{A}} }$, shown in Algorithm~\ref{alg:3}, which employs $\mathcal{M}'^{\mathcal{A}}$ to break the UD property.

\begin{algorithm}[H]    \label{alg:3}
\small
	\DontPrintSemicolon
	\SetKwInOut{Input}{Input}\SetKwInOut{Output}{Output}
	\Input{Security parameter $n$, a sample $(u,k)$.} 
	\Output{0 or 1.}
	Generate ${\bf r}\rand \{0,1\}^{n(w-1})$\;
	Input $n$ and $(\beta^{\star},c_k^{\beta^{\star} - (i^{\star}+1)}(u, {\bf r}_{i^{\star}+1,w-1}),{\bf r},k)$ into $M'^{\mathcal{A}}$\;
	\Return the result from $M'^{\mathcal{A}}$
	\caption{$\mathcal{B}^{ \mathcal{M}'^{\mathcal{A}} }$}
\end{algorithm}

One can see that
\begin{equation} \label{eq:identity}
	{\rm Adv}_{\mathcal{D}_{{\rm UD},\mathcal{U}},\mathcal{D}_{{\rm UD}, \mathcal{F}_n}}(\mathcal{B}^{ \mathcal{M}'^{\mathcal{A}} })=
	{\rm Adv}_{\mathcal{H}_{i^{\star}},\mathcal{H}_{i^{\star}+1}}(\mathcal{M}'^{\mathcal{A}}),
\end{equation}
since the input to $\mathcal{M}'^{\mathcal{A}}$ with $(u,k)$ from $\mathcal{D}_{{\rm UD},\mathcal{U}}$ is equivalent to a sample from $\mathcal{H}_{i^{\star}+1}$, 
while this input to $\mathcal{M}'^{\mathcal{A}}$ with $(u,k)$ from $\mathcal{D}_{{\rm UD}, \mathcal{F}_n}$ is equivalent to a sample from $\mathcal{H}_{i^{\star}}$.
Indeed, 
\begin{equation}
	c_k^{\beta^{\star} - (i^{\star}+1)}(f_{k}(x), {\bf r}_{i^{\star}+1,w-1})=c_k^{\beta^{\star}- i^{\star}}(x\oplus r_{i^{\star}}, {\bf r}_{i^{\star},w-1})
\end{equation}
and $x\oplus r_{i^{\star}}$ is indistinguishable from the uniformly random string.
At the same time, we have 
\begin{equation} \label{eq:insecUD}
	{\rm Adv}_{\mathcal{D}_{{\rm UD},\mathcal{U}},\mathcal{D}_{{\rm UD}, \mathcal{F}_n}}(\mathcal{B}^{ \mathcal{M}'^{\mathcal{A}} }) \leq    {\rm InSec^{UD}}(\mathcal{F}_n;\widetilde{t}).
\end{equation}
The runtime bound $\widetilde{t}=t+3lw+w-2$ is obtained as sum of time $t$ required for $\mathcal{A}$, at most $3lw$ calculations of $f_{k}$ required in  ${\sf Kg}$, ${\sf Sign}$, and ${\sf Vf}$ used in $\mathcal{M}'^{\mathcal{A}}$, and at most $w-2$ calculations of $f_{k}$,
while preparing input for $\mathcal{M}'^{\mathcal{A}}$ in $\mathcal{B}^{ \mathcal{M}'^{\mathcal{A}} }$ (line 2 in Algorithm~\ref{alg:3}) and preparing $\alpha$th chain in $\mathcal{M}'^{\mathcal{A}}$ (line 3 in Algorithm~\ref{alg:2}) (the total number of $f_{k}$ evaluations is given by $w-1-(i^{\star}+1)\leq w-2$).

By combining together Eqs.~\eqref{eq:ineqs}, \eqref{eq:identity}, and \eqref{eq:insecUD} we obtain
\begin{equation}
	{\rm Adv}_{\mathcal{D}_{\mathcal{M}},\mathcal{D}_{{\sf Kg}}}(\mathcal{M}'^{\mathcal{A}})<w\cdot {\rm InSec^{UD}}(\mathcal{F}_n;\widetilde{t}).
\end{equation}
Then putting this result into Eq.~\eqref{eq:epsineq} we arrive at 
\begin{equation}
	\epsilon_{\mathcal{A}} < lw\cdot\left(w \cdot {\rm InSec^{UD}}(\mathcal{F}_n;\widetilde{t})+\widetilde{\epsilon}_{\mathcal{A}} \right)
\end{equation}
Finally, taking into account Eq.~\eqref{eq:Insecs} we obtain the desired upper bound.
\end{proof}

\begin{remark}
One can note that the bound $\widetilde{t}=t+3lw+w-2$ can be tightened at least to
$\widetilde{t}=t+3lw$ by firstly choosing the value $\alpha$ and then removing calculation of $\alpha$th chain within ${\sf Kg}$ used in $\mathcal{M}^{\mathcal{A}}$ and $\mathcal{M}'^{\mathcal{A}}$.
However, it has almost no practical value since usually is assumed that $t\gg4lw$.
\end{remark}

\subsection{Difference from the previous version of the proof}

Here we point out main differences between our security proof and the original proof from Ref.~\cite{Husling} that contains a slightly different security bound, namely:
\begin{multline}\label{eq:huslingproof}
	{\rm InSec}^{\rm EU-CMA}(\text{\rm W-OTS}^+(1^n,w,m);t,1) \\
	\leq 
	wl \cdot \max\left\{
	{\rm InSec^{OW}}(\mathcal{F}_n;t'),
	w\cdot {\rm InSec^{SPR}}(\mathcal{F}_n;t')
	\right\} + \\ w \cdot {\rm InSec^{UD}}(\mathcal{F}_n;t^{\star}) ,
\end{multline}
where $t'=t+3lw$ and $t^{\star}=t+3lw+w-1$.

First of all, during the discussion of $\mathcal{M}^{\mathcal{A}}$, that is the same in both proofs, it was stated that $\Pr[b_\alpha=\beta] \geq \frac{1}{w}$,
motivated by the fact that $\beta$ is chosen uniformly at random (see p. 181 of~\cite{Husling}).
However, as we discussed in our proof, $\mathcal{A}$ may reveal the chain containing challenges, and also may always ask to sign a message with $b_{\alpha}=0$ thus making $\Pr[b_\alpha=\beta]=0$.

In the proof of~\cite{Husling} it is stated that
$\Pr[b'_{\alpha}<\beta| \text{``Forgery is valid''} \wedge b_\alpha=\beta] \geq l^{-1}$.
This is also may not be correct if $\mathcal{A}$ is able to reveal the chain containing challenges and, e.g., make forgery only with $b'_{\alpha}=\beta$.
Actually, accounting a possibility of hostile behavior of $\mathcal{A}$ forces us to introduce the ``Forgery is fortunate'' event and bound its probability by employing ${\rm InSec}^{{\rm UD}}(\ldots)$.
We note that our treatment also gives a different factor before the term ${\rm InSec^{UD}}(\ldots)$.

Moreover, in Ref.~\cite{Husling} the obtained bound contains $\max\{{\rm InSec^{OW}}(\ldots), w{\rm InSec^{SPR}}(\ldots)\}$ instead of ${\rm InSec^{OW}}(\ldots)+w \cdot{\rm InSec^{SPR}}(\ldots)$.
Perhaps, it appeared by putting multiples $p$ and $(1-p)$ on the opposite side of inequalities corresponded to Eq.~\eqref{eq:OWbound} and Eq.~\eqref{eq:SPRbound} of the present paper.

Finally, the used different runtime bounds $t'$ and $t^{\star}$ for breaking OW/SPR and UD of $\mathcal{F}_{n}$, however, as it is shown above they can be considered to be the same.

Anyway, as we demonstrate below, both expressions~\eqref{eq:ourproof} and ~\eqref{eq:huslingproof} provide close levels of security.
Moreover, we note that the security level of W-OTS$^{+}$ used in the security proof of SPHINCS coincides with the derived expression~\eqref{eq:ourproof} (see $\#_{ots}$ term on page 382 of~\cite{SPHINCS}).

\subsection{Security level}

Given results of the Theorem~\ref{thm:1}, we are able to compute the security level against classical and quantum attacks. 
Following reasoning from Refs.~\cite{Husling,Len04}, we say the scheme has security level $b$ if a successful attack is expected to require $2^{b-1}$ evaluations of functions from $\mathcal{F}_{n}$.
We calculate lower bound on $b$ by considering the inequality ${\rm InSec}^{\rm EU-CMA}(\text{\rm W-OTS}^+(1^n,w,m);t,1)\geq 1/2$.
We assume that 
\begin{equation}
	{\rm InSec^{\rm OW}}(\mathcal{F}(n);t)={\rm InSec}^{\rm SPR}(\mathcal{F}(n);t)={\rm InSec}^{\rm UD}(\mathcal{F}(n);t)=\frac{t}{2^n}
\end{equation}
for brute force search attacks with classical computer~\cite{Dods}, and 
\begin{equation}
	{\rm InSec^{\rm OW}}(\mathcal{F}(n);t)={\rm InSec}^{\rm SPR}(\mathcal{F}(n);t)={\rm InSec}^{\rm UD}(\mathcal{F}(n);t)=\frac{t}{2^{n/2}}
\end{equation}
for attack with quantum computer using Grover's algorithm~\cite{Grover}.
We also assume that $t\gg 4lw$, so all runtime bounds used in~\eqref{eq:ourproof} and ~\eqref{eq:huslingproof} are the same: $\widetilde{t}\approx t' \approx t^{\star} \approx t$.
The results of comparison are shown in Table~\ref{tab:1}.
The new bound is smaller the previous one by $\log\frac{l(2w+1)}{lw+1}\approx 1$ bit for typical parameter values $w=16$ and $l=67$.

\begin{table}[h]
	\centering
	\begin{tabular}{c|c|c}
		& Bound from~\cite{Husling} & Bound from present work\\ \hline
		Classical attacks & 
		$b> n-\log w-\log(lw+1)$ & $b>n-\log(lw)-\log(2w+1)$ \\
		Quantum attacks & $b>\frac{n}{2}-\log w-\log(lw+1)$ & $b>\frac{n}{2}-\log(lw)-\log(2w+1)$ \\
	\end{tabular}
	\caption{Comparison of security levels for the W-OTS$^{+}$ scheme.}
	\label{tab:1}
\end{table}

\section{Conclusion and outlook} \label{sec:concl}

Here we summarize the main results of our work.
We have recapped the security analysis of the W-OTS$^{+}$ signature presented in Ref.~\cite{Husling}, and pointed out some of its flaws.
Although the updated security level almost coincides with the one from Ref.~\cite{Husling}, we believe that our contribution is important for a fair justification of the W-OTS$^{+}$ security.

We note that a security analysis of the many-times stateless signature scheme SPHINCS$^{+}$, which uses W-OTS$^{+}$ a basic primitive and which was submitted to NIST process~\cite{SPHINCS+submission}, originally was based on another approach for evaluating the security level~\cite{Mitigating}. 
However, it was discovered that the employed security analysis has some critical flaws~(see C.J. Peikert official comment on Round 1 SPHINCS$^{+}$ submission~\cite{SPHINCS+comment}).

Recently, a new approach for the security analysis of hash-based signature was introduced~\cite{WEBSITE:3}.
It suggests a novel property of hash functions, namely the decisional second-preimage resistance, and therefore requires an additional deep comprehensive study.


\def\refname{References}

\end{document}